\documentclass[12pt]{amsart}
\usepackage[applemac]{inputenc}
\usepackage[titletoc]{appendix}
\usepackage{bm}
\usepackage{blindtext}
\usepackage{amssymb}
\usepackage{amsfonts}
\usepackage{amsthm}
\usepackage{tikz}
\usepackage{tikz-cd}
\usepackage{caption}
\usepackage{color}
\usepackage{graphicx}
\usepackage{xcolor}
\usepackage{shadethm}
\usepackage{subfigure}

\usepackage{epigraph}

\usepackage{enumerate}
\usepackage{verbatim}
\usetikzlibrary{trees}
\usetikzlibrary{decorations.markings}
\usetikzlibrary{arrows}
\usepackage{mathrsfs}
\usepackage[margin=1.0in]{geometry}
\usepackage{xparse}

\usetikzlibrary{positioning,arrows,patterns}
\usetikzlibrary{decorations.markings}
\usetikzlibrary{calc}

\NewDocumentCommand\semiloop{O{black}mmmO{}O{above}}
{%
\draw[#1] let \p1 = ($(#3)-(#2)$) in (#3) arc (#4:({#4+180}):({0.5*veclen(\x1,\y1)})node[midway, #6] {#5};)
}

\newcommand{\Ker} {\mathrm{ker}}
\newcommand{\w} {\omega}
\newcommand{\TM}          {T^*M}
\newcommand{\Lie}         {\mathcal L}

\usepackage{url}
\usepackage{hyperref}
\hypersetup{colorlinks=false, linktocpage, linkbordercolor=red, citebordercolor=green, filebordercolor=red, urlbordercolor=cyan, pdfborderstyle={/S/U/W 1}}

\theoremstyle{plain}
\newtheorem{thm}{Theorem}[section]
\newtheorem{lem}[thm]{Lemma}
\newtheorem{rmk}[thm]{Remark}

\newtheorem{prop}[thm]{Proposition}

\newtheorem{defn}{Definition}[section]

\theoremstyle{remark}
\newtheorem{rem}{Remark}[section]

\newtheorem{ex}{Example}[section]

\newcommand{\pth}         {\mathcal{P}}

\def\gpd{\,\lower1pt\hbox{$\longrightarrow$}\hskip-.24in\raise2pt
               \hbox{$\longrightarrow$}\,}

\newcommand{\F}{\mathcal{F}}

\makeatletter

\pgfdeclareshape{genus}{
 \anchor{center}{\pgfpointorigin}
\backgroundpath{
    \begingroup
    \tikz@mode
    \iftikz@mode@fill

         \pgfpathmoveto{\pgfqpoint{-0.78cm}{-.17cm}}
     \pgfpathcurveto %
            {\pgfpoint{-0.35cm}{-.44cm}}
        {\pgfpoint{0.35cm}{-.44cm}}
        {\pgfpoint{.78cm}{-0.17cm}} 
     \pgfpathmoveto{\pgfqpoint{-0.78cm}{-0.17cm}}
     \pgfpathcurveto %
            {\pgfpoint{-0.25cm}{.25cm}}
        {\pgfpoint{.25cm}{.25cm}}
        {\pgfpoint{0.78cm}{-0.17cm}}
        \pgfusepath{fill}
        \fi

        \iftikz@mode@draw
     \pgfpathmoveto{\pgfqpoint{-1cm}{0cm}}
     \pgfpathcurveto %
            {\pgfpoint{-0.5cm}{-.5cm}}
        {\pgfpoint{0.5cm}{-.5cm}}
        {\pgfpoint{1cm}{0cm}}

     \pgfpathmoveto{\pgfqpoint{-0.75cm}{-0.15cm}}
     \pgfpathcurveto %
            {\pgfpoint{-0.25cm}{.25cm}}
        {\pgfpoint{.25cm}{.25cm}}
        {\pgfpoint{0.75cm}{-0.15cm}}
              \pgfusepath{stroke}
        \fi
        \endgroup
    }
    }

    \makeatother

\begin{document}

\title{Poly-symplectic geometry and the AKSZ formalism}
\author[I. Contreras]{Ivan Contreras}
\author[N. Martinez]{Nicolas Martinez Alba}
\email[I.~Contreras]{icontreraspalacios@amherst.edu}
\email[N. Martinez]{nmartineza@unal.edu.co}
\maketitle

\begin{abstract}
In this paper we extend the AKSZ formulation of the Poisson sigma model to more general target spaces, and we develop the general theory of graded geometry for poly-symplectic and poly-Poisson structures. In particular, we prove a Schwarz-type theorem and transgression for graded poly-symplectic structures, recovering the action functional and the poly-symplectic structure of the reduced phase space of the poly-Poisson sigma model, from the AKSZ construction.

\end{abstract}


\section{Introduction}
A covariant formulation of Hamiltonian field theory goes back to the works of de Donder \cite{deDonder} and Carath\'eodory \cite{Caratheodory}. In this formulation, the cotangent bundle as the phase space is replaced with the co-jet bundle, which comes equipped with a canonical poly-symplectic structure. This structure is a natural extension of the Liouville form of the cotangent bundle $T^*M$ to the Whitney sum $\oplus_{i=1}^k T^*M$. This covariant formulation has proven to be successful in describing certain classes of field theories, such as systems of $N$ classical particles in terms of  1-dimensional field theory. This notion of vector bundle --valued symplectic forms also leads to the definition of poly--Poisson structures \cite{IMV,NMA-lmp}, analogous to the way in which Poisson structures generalize symplectic forms.
\newline
Motivated by the properties of poly--Poisson structures and their similarities with the usual Poisson bivectors, we studied the {\it poly--Poisson sigma model} in \cite{PPSM}. This model extends naturally the Poisson sigma model, a 2-dimensional field theory relevant to the study of deformation quantization of Poisson structures \cite{CattaneoFelder2}, as well as the integration problem of Lie algebroids \cite{CattaneoFelder}. The poly--Poisson sigma model was constructed by replacing the  Poisson target by a poly-Poisson structure $(S,P)$, and using the fact that the theory is topological to describe a geometric structure of the phase space that integrates a given poly-Poisson structure $(S,P)$ . In particular, it is proven \cite{PPSM} that the reduced phase space can be equipped with a poly-symplectic groupoid structure, extending the results by Cattaneo and Felder \cite{CattaneoFelder}, as well as Cattaneo and the first author \cite{CattaneoContreras}, to the setting of poly-symplectic and poly-Poisson structures. 
\newline
The extensions of such results are global and canonical in nature, independent on the choice of coordinates. This particular fact leads to a generalization of this construction by extending the usual Poisson and symplectic structures to the setting of graded differential geometry. Supermanifolds were originally introduced in the physics literature to provide a
geometric framework  to the emerging theory of supersymmetry. A subsequent integral degree was considered to include Poisson structures, Lie and Courant algebroids.
\newline
In this note we consider a particular case of this construction: poly-symplectic forms in graded manifolds and their transgression, following the AKSZ construction as in \cite{AKSZ, CattaneoFelder3}.
We promote the poly-momentum formalism to topological field theories of AKSZ type. In particular, we prove a poly-symplectic analogue of the Schwarz theorem (Theorem~\ref{thm:poly-Schwarz}) which provides a Darboux-type result for graded poly-symplectic manifolds. The other key result is the existence of transgression (Theorem~\ref{thm:transgerssion2}) for mapping spaces with poly-symplectic target. Last but no least, we describe the classical master equation in the graded poly-symplectic formalism (Theorem~\ref{thm:polyCME})  revisiting the vector-valued version of the Poisson sigma model in this framework. 
\newline
This construction serves as a first step toward the perturbative quantization of the poly-Poisson sigma model, where the goal is to describe and to solve the deformation quantization problem for poly-Poisson structures by using the corresponding sigma model as an AKSZ theory. 
It also provides a bridge between topological field theories of AKSZ type and the covariant formulation of field theory via poly-symplectic geometry \cite{Caratheodory, Gunther}. A vector field-formulation of AKSZ theories provides interesting examples supported by vector-valued variational principles \cite{Finet}, as well as a covariant-differential formulation of Lagrangian field theory \cite{Canarutto}.
\subsection*{Acknowledgements} I.C. would like to thank Alberto Cattaneo and Michele Schiavina for insightful discussions.  We also thank the anonymous referees for their valuable comments and suggestions, that greatly improved the content of the paper.
\section{Overview of poly-symplectic and poly-Poisson geometry}\label{sec:PP}
\begin{defn}\label{Def:PolyP}
A \textbf{Poly-Poisson  structure} of order $r$, or simply an $r$-Poisson structure, on a manifold $M$ is a pair $(S,P)$, where $S\to M$ is a vector sub-bundle of $\TM\otimes \mathbb R^r$ and $P:S\to TM$ is a vector-bundle morphism (covering the  identity) such that the following conditions hold:

\begin{itemize}
\item[(i)] $i_{P(\eta)} \eta=0$, for all $\eta\in S$,
\item[(ii)] $S^\circ= \{X\in TM| i_X\eta=0, \,\forall\,\eta\in S\}=\{0\}$,
\item[(iii)] the space of section $\Gamma(S)$ is closed under the bracket
\begin{equation}\label{Def:bracketHP} 
\lfloor \eta,\gamma \rfloor:=\Lie_{P(\eta)}\gamma-i_{P(\gamma)}d\eta \; \mbox{ for } \;\gamma,\eta\in \Gamma(S),
\end{equation}
and the restriction of this bracket to $\Gamma(S)$ satisfies the Jacobi identity.
\end{itemize}
We will call the triple $(M,S,P)$ an $r$-\textbf{Poisson manifold}. 
\end{defn}
\begin{ex}[\textbf{Poly-symplectic structures}]\label{ex:PS}
First we will consider the case of $P$ being an isomorphism of vector bundles. By using the natural projections $p_j:S\to \TM$ we can define the following bundle map $$\omega_j^{\sharp}:TM\overset{P^{-1}}{\longrightarrow}S\to T^*M.$$ Condition (i) in Def.~\ref{Def:PolyP} is the same as that each $\w_j$ is skew-symmetric, whereas condition (ii) means that $\cap \mbox{ker} \omega_j=0$. Finally, condition (iii) is equivalent to $d\omega_j=0$ for $j=1,\ldots,r$. 
\end{ex}
Similar to the case of Poisson structures, we can define the space of admissible sections 
$$C_{\mathrm{adm}}^\infty(M,\mathbb R^r)=\{\alpha\in C^\infty(M,\mathbb R^r) \vert d\alpha\in \Gamma(S)\}$$
on a poly--Poisson manifold $(M,S,P)$ and define the bracket 
\begin{equation}\label{eq:Poisson-bracket}
\{\cdot,\cdot\}:C_{\mathrm{adm}}^\infty(M,\mathbb R^r)\times C_{\mathrm{adm}}^\infty(M,\mathbb R^r)\to 
C^\infty(M,\mathbb R^r)
\end{equation}
by the relation $\{\alpha,\beta\}:=i_{P(d\alpha)}d\beta$.\\

For any two admissible sections $\alpha,\beta\in C_{\mathrm{adm}}^\infty(M,\mathbb R^r)$, the definition of the bracket \eqref{Def:bracketHP} shows that $d\{\alpha,\beta\}=\lfloor d\alpha,d\beta \rfloor$, hence 
$$\{C_{\mathrm{adm}}^\infty(M,\mathbb R^r),C_{\mathrm{adm}}^\infty(M,\mathbb R^r)\}\subset C_{\mathrm{adm}}^\infty(M,\mathbb R^r)$$ and $[P(d\alpha),P(d\beta)]=P(d\{\alpha,\beta\})$. Finally, the skew-symmetry of $P$ ( item (i) in Definition \ref{Def:PolyP}), leads to the identity $i_{P(d\{\alpha,\beta\})}d\gamma+i_{P(d\gamma)}d\{\alpha,\beta\}=0$ for any three admissible sections, and this implies that 
$$\{\alpha,\{\beta,\gamma\}\}+\{\beta,\{\gamma,\alpha\}\}+\{\gamma,\{\alpha,\beta\}\}=0.$$

\subsection{The poly-Poisson Sigma Model (PPSM)}
The Poisson Sigma Model has been studied extensively in the literature (e.g. \cite{CattaneoContreras,CattaneoFelder, CattaneoFelder2, CattaneoFelder3, IvanThesis}), from the point of view of deformation quantization, as well as the integration of Poisson structures. 
Following \cite {PPSM}, a 2-dimensional Topological Field Theory (TFT) can be constructed using the following data:
\begin{enumerate}
    \item \emph{Source space:} A two dimensional disk $\Sigma=D^2$, with the boundary $\partial \Sigma$ split into closed
intervals $I$ intersecting at the end points. On alternating intervals we impose boundary conditions in such a way that the remaining intervals are free (following \cite{ IvanThesis, PPSM})\footnote{For the usual Poisson Sigma Model, this choice of boundary conditions and sectors allows to construct the symplectic groupoid integrating a given (integrable) Poisson manifold, via the reduced phase space of the theory.}.
    \item \emph{Target space:} A poly-Poisson structure $(S,P)$.
\end{enumerate}
In this TFT the space of bulk fields for PPSM is given by
$\F^{PP}=\mbox{Map}(T\Sigma, T^*M\otimes \mathbb R^r) 
$ whereas the space of boundary fields is 
$
\F^{PP}_{\partial}=\mbox{Map}(T\partial \Sigma, T^*M \otimes \mathbb R^r)
$.
Using the same boundary conditions as in the PSM, we can identify $\F^{PP}_{\partial}$ with the $r$-th Whithey sum of the cotangent bundles of the path-space of $M$. Therefore
\[\F^{PP}_{\partial}:=\oplus_r T^*(\pth(M))\cong \pth(\oplus_r T^*M),\] from which we can deduce the following result:

\begin{thm}\label{thm:PPSM}
$\F^{PP}_{\partial}$ is a weak poly-symplectic Banach manifold. Moreover, the submanifold $\F(S,P)=\pth(S)\overset{\iota}{\to}\F^{PP}_{\partial}$ is a weak $r$-symplectic submanifold with the restriction of the graded symplectic form in $\F^{PP}_{\partial}$.
\end{thm} 
The proof can be found in \cite[Section~IV.B]{PPSM}, however we give a short comment on the proof of the second statement. First consider any element $\gamma\in \pth(S)$ and  denote by $\gamma_b$ its base path and $(\gamma^1,\dots,\gamma^r)\in S_{\gamma_b}$. If $\bar{\gamma},\bar{\eta}$ are two curves in $\pth(S)$ starting at $\gamma$ and $\delta \bar{\gamma},\delta \bar{\eta}$ represent tangent vectors of paths, then 
$$\w(\delta \bar{\gamma},\delta \bar{\eta})=\int_0^1(\dots,\delta_1\bar{\gamma}_b\delta_2\bar{\eta}^i-\delta_1\bar{\eta}_b\delta_2
\bar{\gamma}^i,\dots)_{i=1,\dots,r}dt.$$
We now assume that $\delta \bar{\gamma}\in \Ker \omega$, $\bar{\eta}$ has constant base path and is linear on the fibers. We obtain that
$$\int_0^1(\delta_1\bar{\gamma}_b\bar{\eta}^1,\dots,\delta_1\bar{\gamma}_b\bar{\eta}^r)dt=0$$
for any $(\bar{\eta}^1,\dots,\bar{\eta}^r)\in S_{\bar{\eta}_b}$, hence $\delta_1\bar{\gamma}_b\in S_{\bar{\eta}_b}^0=\{0\}$. This leads us to note that, for any curve $\bar{\eta}$ in $\pth(S)$, the following holds:
$$\w(\delta \bar{\gamma},\delta\bar{\eta})=\int_0^1(-\delta_1\bar{\eta}_b\delta_2\bar{\gamma}^1,\dots,-\delta_1\bar{\eta}_b\delta_2\bar{\gamma}^r)dt=0.$$
As $\delta_1\bar{\eta}$ runs over $TM$, we can conclude that $\delta_2\bar{\gamma}^i=0$, which completes the assertion that $\delta\bar{\gamma}=0$, when $\delta\bar{\gamma}\in \Ker\w$.

Finally, the integration procedure via the poly--Poisson sigma model needs an action that is described by
\begin{equation}\label{PPSM_Action}
\int_{\Sigma} \langle \eta, dX \rangle + \frac 1 2 \langle P(X)\eta, \eta \rangle,
\end{equation}
where $(X,\eta)$ are coordinates of  $\F^{PP}$ and the pairing takes values in the vector space $\mathbb{R}^r$.

The reduced space yielding the integration comes from the local gauge symmetries for PPSM, $\mathfrak g_0$ as the (infinite-dimensional) Lie algebra of $\mathcal C^1$-maps $\alpha: [0,1] \to \Gamma(S)$ such that
\begin{enumerate}
\item $\alpha(0)=\alpha(1)=0, \forall \alpha \in \Gamma(S)$.
\item The Lie bracket is given by $ [ \alpha, \beta](t)= \lfloor \alpha(t),\beta(t) \rfloor$, i.e. the bracket is defined pointwise by using the bracket $\lfloor \cdot, \cdot \rfloor$ from the sections on $S$.
\end{enumerate}
By using the $r$-Poisson structure of $(S,P)$ in the definition of $\mathfrak g_0$ we can define two operations:
\begin{align}\label{def:g-PPaction} 
\xi:\mathfrak g_0\to \mathfrak{X}(\F^{PP}_{\partial}); \hspace{1cm} 
&\xi_{\beta}(X)=-P_X(\beta)\\ \notag
&\xi_{\beta}(\eta) = d\beta + \partial P_x\eta \beta\\ \label{def:PPm.m.}
H:\mathfrak g_0\times \F^{PP}_{\partial}\to \mathbb R^r ; \hspace{1cm}
&(\beta,(X,\eta))\mapsto \int_0^1\beta(dX) dt-\int_0^1\eta(P_X(\beta))dt
\end{align}
It follows  \cite{PPSM} that $\mathfrak g_0$ generates the infinitesimal symmetries of PPSM, as described in the following
\begin{thm}
Let $\beta$ be an element of $\mathfrak g_0$.
\begin{enumerate}
\item The action by $\beta$ is a lift of an action by $\beta$ on $\mbox{Map} ([0,1], M)$.
\item The action is Hamiltonian with respect to the weak $r$-symplectic and Hamiltonian function $H$, i.e. $i_{\zeta(\beta)}\w=d\langle H,\beta\rangle$ with the $\mathbb{R}^r$-valued pairing.
\end{enumerate}
\end{thm}

\begin{rem}
The reduction of the space of $r$--covelocities (cf. Example 2.17~\cite{PPSM}) can be extended to the infinite dimensional case yielding a reducible weak $r$--symplectic action, hence a weak $r$--symplectic reduction in degree 0 can be done. Furthermore, using such reduction we also obtain that the coisotropic reduction inherits an $r$--symplectic groupoid structure integrating the Lie algebroid $(S, P)$. For more details on these procedures we refer the reader to Section 4 of \cite{PPSM}.
\end{rem}


\section{Graded manifolds and the AKSZ formalism}
The following definition of graded manifold, suitable for the AKSZ formalism, is based on \cite{CattaneoSchaetz}. The grading is assumed to be nonnegative. This definition includes infinite dimensional examples, e.g. mapping spaces. For more details on a rigorous formulation, see \cite{CattaneoFelder3, CattaneoSchaetz}.
\begin{defn} A graded manifold $\mathcal M$ is a locally ringed space $(M, \mathcal O_M)$ that is locally isomorphic to the $\mathbb Z$-graded algebra \footnote{On a $\mathbb Z$-graded manifold there is a  natural $\mathbb Z_2$-grading (which defines even and odd coordinates) given by the parity of the $\mathbb Z$-grading.}
\[(U, \mathcal C^{\infty} (U)\otimes \wedge W^*),\]
where $U$ is an open subset of $\mathbb R^n$ and $W$ is a finite dimensional vector space.
\end{defn}
Any graded manifold has a {\it global} vector field $E$, called the \emph{Euler vector field}, that allows us to differentiate in the direction of the base structure. In coordinates we can write $E=p_i\partial p_i$.
\begin{ex}\label{ex:OddTangent}(Odd tangent bundle). If $M$ is an ordinary manifold, the graded manifold $T[1]M$ has as sheaf of functions the de Rham complex $\Omega(M)$.
\end{ex}
\begin{ex}\label{ex:oddCT}(Odd cotangent bundle). If $M$ is an ordinary manifold, the graded manifold $T^*[1]M$ has as sheaf of functions the Gerstenhaber algebra of poly-vector fields.
\end{ex} 
\begin{rmk}
A similar construction described in the previous examples works the same when $\mathcal M$ is a graded manifold.
\end{rmk}

An important set of objects in this formalism is the set of 1-graded symplectic forms $\w$ on graded manifold $M$, which correspond to homogeneous 2-forms, i.e $\Lie_E\w=\w$, of degree 1, closed with respect to the de Rham differential, and such that $\w^\sharp:T\mathcal M\to T^*[1]\mathcal M$ is an isomorphism of (graded) vector bundles. 

The graded manifold in the example \ref{ex:oddCT}, $T^*[1]\mathcal M$ (where $\mathcal M$ is either graded or ungraded) always comes equipped with a symplectic structure that resembles the canonical symplectic form of the ordinary cotangent bundle. Furthermore, Schwarz' Theorem states that any degree 1 symplectic graded manifolds is always of the form $T^*[1]M$ with an ordinary manifold $M$.
\begin{thm}\label{Schwarz-Roytenberg}\cite{Roy}
There is a 1-to-1 correspondence between nonnegatively -graded symplectic manifolds of degree 1 and shifted cotangent bundles of smooth manifolds.
\end{thm}
Note that this result is a global version of the Darboux theorem for ordinary symplectic manifolds.  In addition to a graded symplectic structure, there is a differential defined on $\mathcal C^{\infty}(\mathcal M)$, called the \emph{cohomological vector field}. We say that a graded vector field $X$ on a graded manifold $\mathcal M$ is \emph{cohomological} if 
\begin{enumerate}
    \item $X$ has degree +1,
    \item $[X,X]=0$.
\end{enumerate}
\begin{ex}
Following Example \ref{ex:OddTangent}, the graded manifold $T[1]M$ has the de Rham differential $X=d_{DR}$ as a cohomological vector field.
\end{ex}
\begin{ex}
Following-up with Example \ref{ex:oddCT}, a cohomological vector field is given by $X=[\Pi, \bullet]_{SN}$, where $\Pi$ is a Poisson vector field on $M$. Condition (2) is equivalent to the Schouten-Nijenhuis equation $[\Pi,\Pi]_{SN}=0$.
\end{ex}
The previous example is a particular instance of a $QP$\emph{-manifold}, which is a graded manifold equipped with both a cohomological vector field $Q$, and a graded symplectic structure $P$.

\subsection{The AKSZ formalism}
In this section we are interested in  describing the geometry on the space of smooth maps between two graded manifolds. A key ingredient for this theory is the \emph{Berezinian integration} \cite{CattaneoSchaetz, Manin} , which allows us to extend the Lebesgue measure to the graded setting. In particular, given two graded manifolds  $\mathcal N$ and $\mathcal M$, and a measure $\mu$ on $\mathcal N$, we want to push-forward $\mu$ at the level of differential forms. More precisely, there is a chain map $\mu_*: \Omega(\mathcal N\times \mathcal M)\to \Omega(\mathcal M)$ defined by 
\[(\mu_* \omega(z)(\lambda_1, \lambda_2, \cdots, \lambda_k))=\int_{y\in N}\omega(y,z)(\lambda_1, \lambda_2, \cdots, \lambda_k)\mu(y),\]
where $z \in \mathcal M$ and $\lambda_1, \lambda_2,\cdots, \lambda_k \in T_z \mathcal M$.
This construction allows us to push forward along the source manifold $N$, taking care of the odd directions. The usual Lebesgue integration is enhanced with the integration rule along odd coordinates:
\begin{equation}\label{Odd Berezinian}
\int \xi d\xi=1.
\end{equation}
For instance, following up on Example \ref{ex:OddTangent}, if $\mathcal N=T[1]X$  and $f \in \mathrm{Fun}(\mathcal N)$, then the Berezinian integration of $f$ is 
\[\int_{\mathcal N} \, f= \int_X j(f),\]
where $j$ is the identification between functions on $\mathcal N$ and $\Omega(X)$.

The key role of Berezinian that we want to use in these notes, is the induced map on the space $\mathcal{F}=\mathrm{Map}(\mathcal N,\mathcal M)$. For the desired map we first consider the canonical projection and evaluation maps 
$$P:N\times \mathcal{F}(\mathcal N,\mathcal M)\to \mathcal{F}(\mathcal N,\mathcal M) \mbox{\ \ and\ \ } \mathrm{ev}:\mathcal N\times \mathcal{F}(\mathcal N,\mathcal M)\to \mathcal M.$$ 
From these two maps we can consider the composition of pull-back and push-forward that lead us to construct differential forms on mapping spaces. Formally, we can define the map between the space of graded differential forms as follows:
\begin{eqnarray*}
\mathbb T:\Omega(\mathcal M) &\to& \Omega(\mathcal{F}(\mathcal N,\mathcal M))\\
\alpha &\mapsto& P_{*}\mbox{ev}^*(\alpha)= \int_{\mathcal N} \, \mbox{ev}^*(\alpha).
\end{eqnarray*}
In other words, we can identify the previous map with the induced $\mu_*$ operation, i.e. $\mathbb T=\mu_*\mathrm{ev}^*$ with $\mu_*$ as before. 

\subsection{Transgression on the mapping space}
Let $Q_{\mathcal N}$ and $Q_{\mathcal M}$ be two cohomological vector fields on $\mathcal N$ and $\mathcal M$, respectively. Using the canonical identification $T_f\mbox{Map}(\mathcal N,\mathcal M)\simeq \Gamma(\mathcal N,f^*T\mathcal M)$ for any smooth function $f:\mathcal N\to \mathcal M$, we can lift the vector fields $Q_{\mathcal N}$ and $Q_{\mathcal M}$ to $\mbox{Map}(\mathcal N,\mathcal M)$ as follows:  
\begin{align*}
    \hat{Q}_{\mathcal N}(f):& \mathcal N\to f^*T\mathcal M\\
    &x\mapsto d_xfQ_{\mathcal N}(x)\\
\breve{Q}_{\mathcal M}(f):& \mathcal N\to f^*T\mathcal M\\
&x\mapsto Q_{\mathcal M}(f(x)).
\end{align*}

And, as a general fact, if both vector fields are self-commuting then any linear combination is self commuting (see \cite{CattaneoFelder3}, Sec~2.2.1).

The following theorem \cite{CattaneoFelder3} guarantees the existence of a symplectic structure on the mapping space, provided that the target space is symplectic.
\begin{thm}\label{AKSZ-Transgression}
If we choose a nondegenerate measure $\mu$ on $\mathcal N$, the transgression of a symplectic structure $\omega$ on $\mathcal M$ induces a symplectic structure $\Omega$ on $\mathrm{Map}(\mathcal N,\mathcal M)$ defined by 
\begin{equation}
\Omega= \mathbb T \omega.    
\end{equation} 
\end{thm}

\subsection{The classical master equation}
The following theorem \cite{AKSZ, CattaneoFelder3} states the existence, under certain conditions, of a function $\mathcal S$ that satisfies the so called \textbf{classical master equation} (CME). See Subsection \ref{Poly-CME} for a version of this theorem in the poly-symplectic formalism (Theorem \ref{PCME}).
\begin{thm}\label{thm:CME}
Let $Q_{\mathcal N}$ and $Q_{\mathcal M}$ be two cohomological vector fields on $\mathcal N$ and $\mathcal M$ respectively, let $\mu$ be a $Q_{\mathcal N}$-invariant measure on $\mathcal N$ and let $\omega$ be a symplectic structure on $\mathcal M$. If $\omega$ is exact, i.e. $\omega= d \theta$, then $\hat{Q}_{\mathcal N}$ is a Hamiltonian vector field for the function $\hat{S}=-\iota_{\hat{Q}_{\mathcal N}} \Theta$, where 
$\Theta= \mu_{*} \mbox{ev}^* \theta.$
Furthermore, if  $Q_M$ admits a Hamiltonian function $S_{\mathcal M}$, then the function \[\mathcal S:= \mu_{*} \mathrm{ev}^* S_{\mathcal M}+ \breve{S}\] satisfies the \textbf{classical master equation}
\begin{equation}\label{CME}
(\mathcal S,\mathcal S)=0, 
\end{equation}
with respect to the Gerstenhaber bracket $(\cdot,\cdot)$ induced by $\Omega$.
\end{thm}
\section{The AKSZ formalism in the poly-symplectic case}
\subsection{Poly-symplectic structures and $QP$-manifolds}
Similar to the usual symplectic case, we can define a {\bf $k$-graded poly-symplectic} manifold as a nonnegatively-graded manifold
endowed with a closed 2-form $\w$ such that $\w^\sharp:T\mathcal{M}\to \oplus_rT^*[k]\mathcal{M}$ is non-degenerated and homogeneous  with respect to the Euler vector field, i.e $\Lie_E\w=\w$. One of the main examples is the Whitney sum of graded poly-symplectic manifolds. For any usual manifold $M$ we know that $T^*[1]M$ is 1-graded symplectic $N$-manifold. Moreover, it is an easy exercise to verify that $\oplus_rT^*[1]M$ is also an $N$-manifold and that the following facts hold:
\begin{prop}
The canonical projections $\pi_j:\oplus_rT^*[1]M\to T^*[1]M$ are surjective submersions and the common kernel of $d\pi_j$ is trivial.
\end{prop}
By using the previous statement we can prove that $\oplus_rT^*[1]M$ is a 1-graded poly-symplectic manifold (cf. \cite{PPSM}) with 1-degree poly-symplectic form 
\begin{equation}\label{eq:w in T*[1]M}
\w=\oplus_j \pi_j^*\w_{can}=(\w_{can},\overset{r}{\dots},\w_{can}).
\end{equation}
Note that in this case, the induced even and odd coordinates $p_j$ and $q_j^l$ in the Whitney sum from the respective graded coordinates $ (p_j,q_j) $ on each $ T^*[1]M $ yield
$$\w=\sum_{j,l}dq_j^ldp_j.$$ 
In general it is not expected that general 1-graded poly-symplectic graded manifolds have a set of local coordinates, due to the fact that in standard poly-symplectic manifolds such coordinates exists under additional assumptions (see for example \cite{FoGo,NMA}). However, we have the following set of local coordinates:

\begin{prop}\label{prop:w in coord}
For any poly-symplectic graded manifold $(\mathcal{M},\w)$ with odd $p_j$ and even $q_l$ coordinates we get that $$\w=\sum_{j=1}^m\sum_{l=1}^sc_{jl}dq_ldp_j$$ where the $c_{jl}$ are constant functions.
\end{prop}

\begin{proof}
The condition that $\w$ preserves the grading yields that $\w=\sum_{j,l}c_{jl}dq_ldp_j$ for some smooth function $c_{jl}$. As the relation $\w=d(i_E\w)$ holds, we conclude that $c_{jl}$ are constant.
\end{proof}
We are interested in a particular case of such constant functions $c_{jl}$. We call the $N$-manifold $(\mathcal{M},\w)$ {\bf exact $r$-poly-symplectic} if $s=rm$ and the matrix $(c_{jl})_{m\times s}$ is equivalent to $ (\mathrm{Id}_{m\times m}\overset{{\small r}}{\cdots} \mathrm{Id}_{m\times m})_{m\times s} $.
\\
In Section~3. of \cite{Roy} the author presents a sketch of the proof of Schwarz's theorem in terms of the graded Poisson bracket. Here we will describe the idea of the same proof but in terms of the graded differential forms. To do that, we first note that if $ (M,\w) $ is 1-graded symplectic $ N $-manifold then $ \mathcal{M} $ is a vector bundle $ A\to M $, so 
\[ \w^\sharp|_{pt} :A\times TM|_{pt}\to A^*[1]\times T^*M|_{pt}.\]
The condition of $ \w $ being grading-preserving implies that odd coordinates are carried out to even coordinates and vice-versa via $ \w^\sharp $. The homogeneous condition and no-degeneracy of the symplectic form says that $ TM\cong A^*[1]$. Finally, the Leibniz rule from vector fields and the Jacobi identity from closedness of $ \w $ imply that the previous identification induces a symplectomorphism $ \mathcal{M}\cong T^*[1]M $.

It is worth to mention that this is a canonical argument so it can be adapted to the poly-symplectic case. Indeed, when we consider the identification of $ \mathcal{M} $ with $ A\to M $ (which just comes from the 1-degree shift cf. \cite{Roy}), we need an additional ``dimensional" condition guaranteeing that \[ A|_{pt}\cong \oplus_r T_{pt}^*M \] and the remainder of the proof follows the same ideas above. A straightforward verification gives that such additional condition is just $(\mathcal{M},\w)$ being exact $r$-poly-symplectic. Thus we have proved the following
\begin{thm}[1-shifted poly-symplectic Schwarz's Theorem]\label{thm:poly-Schwarz}
For an exact, degree 1, $ r $-poly-symplectic manifold $(\mathcal{M},\w)$ there exists a poly-symplectomorphism $ \mathcal{M}\cong \oplus_r T^*[1]M $ for $ M $ a standard manifold.
\end{thm}

\subsection{Transgression}
This section is devoted to induce a poly-symplectic structure on the space of maps when the target space is endowed with a graded poly-symplectic structure. The first step towards achieving this goal is to consider the graded poly-symplectic structure on the graded manifold $\oplus_rT^*[1]M$, with its canonical form as in \eqref{eq:w in T*[1]M}. Following the same ideas, but in the graded case, as in  Theorem~\ref{thm:PPSM} (see also Proposition~4.3 in \cite{PPSM}), we can prove our first transgression result:
\begin{prop}
For any two manifold $ \Sigma, M $ with boundary $\partial \Sigma$ we get that $$ \mathcal{F}=Map(T[1]\partial \Sigma,\oplus_rT^*[1]M) $$ is a weak poly-symplectic graded manifold of degree 0.
\end{prop}\label{thm:transgerssion1}

\begin{proof}
Recall that Equation~\eqref{eq:w in T*[1]M} gives the graded poly-symplectic form on $ \oplus_rT^*[1]M $. By the same construction, but now using transgression, we can define 
\[ \Omega=(\pi_1^*\w_{st},\ldots,\pi_r^*\w_{st}) \] where $ \w_{st}=\int_{T[1]\partial \Sigma}\delta q\delta p $. It is a straightforward computation to verify that $ \Omega  $ is closed. For the claim about the kernel we should consider constant base paths  and fiberwise linear paths. Those type of paths led us to study the kernel of $ \Omega $ via the kernel of $ \w $ in Equation~\eqref{eq:w in T*[1]M} which is trivial.
\end{proof}

The same result works for an exact 1-degree poly-symplectic manifold, but in this case the integration on fibers is replaced by the  integration with respect to a Berezinian. This consideration yields our main transgression theorem:
\begin{thm}\label{thm:transgerssion2}
 The space $ \mathcal{F}=Map(\mathcal{N},\mathcal{M}) $ is weak-poly-symplectic graded manifold  of degree $1-k$ with cohomological vector field $Q_{total}$, where $(\mathcal{M},\w)$ is an exact 1-degree $ r $-poly-symplectic manifold  and $ (\mathcal{N},Q_\mathcal{N},\mu) $ is a $ Q $--manifold with cohomological vector field $ Q_\mathcal{N} $, with Berezinian $ \mu $ and dimension $k$ (in degree 0) \footnote{In most of the examples regarding AKSZ theories, $\mathcal N\cong T[l]N$, so $k= \dim (N)$.}.
	 
\end{thm}

\begin{proof}
The idea is to adapt the same proof but considering the evaluation $ ev $ and projection $ P $ maps from $ \mathcal{N}\times \mathcal{F} $ to $ \mathcal{M} $ and $ \mathcal{F} $ respectively. Such maps define a graded morphism $ \mathbb{T}:=P_*ev^*: \Omega^q(\mathcal{M})\to \Omega^q(\mathcal{F})$ that, in coordinates, yields \[ \Omega:=\mathbb{T}\w=\int_{\mathcal{N}}\sum_{j,l}\delta q_j^l \delta p_j d\mu, \] via the poly-symplectomorphism as in the Schwarz theorem \ref{thm:poly-Schwarz}. As in the previous proof, the only condition we should verify is that the kernel of $ \Omega $ is trivial, but this lies on the kernel of $ \w=\sum_{j,l}d q_j^l d p_j $ by using constant base maps and odd-wise linear maps on the Berezinian integration.

It only remains to construct the cohomological vector field, but this comes from the lifting of the corresponding cohomological vector fields, that is 
\[ Q_{total}=Q_{\mathcal{N}}^{\mathrm{lift}}+Q_{\mathcal{M}}^{\mathrm{lift}}.\]
\end{proof}

\subsection{The Classical Master Equation: The poly-Poisson Sigma Model Revisited }\label{Poly-CME}
This final section deals with the study of the poly--Poisson version of the classical master equation,  which was introduced in Proposition~\ref{thm:CME}, and how this result explains the choice of the action functional in \cite{PPSM}, for the poly--Poisson Sigma Model. The ingredients for the classical master equation are the transgression of poly--Poisson structures and a well defined {\it Poisson-type} bracket on {\it admissible} functions. The transgression result in Theorem~\ref{thm:transgerssion2} and the bracket in \eqref{eq:Poisson-bracket} yield the following result:

\begin{thm}[Poly--Poisson CME]\label{PCME}
Consider the same conditions as in the Theorem~\ref{thm:transgerssion2}, i.e. let $(\mathcal{M},\Omega)$ be an exact 1-degree $ r $-poly-symplectic manifold with basic form $\theta$, i.e $\Omega=d\theta$, and let $ (\mathcal{N},Q_\mathcal{N},\mu) $ be a $ Q $-manifold with cohomological vector field $ Q_\mathcal{N} $ and Berezinian $ \mu $. If the Berezinian $\mu$ is $Q_\mathcal{N}$-invariant, then the space $(\mathcal{F}=\mathrm{Map}(\mathcal{N},\mathcal{M}),\Omega=\mathbb{T}\omega)$ satisfies the following:
\begin{enumerate}
    \item $\hat{Q}_N$ is Hamiltonian vector field with Hamiltonian function, i.e. if $\Theta= \mathbb T \theta$ we get \[\hat{S}=-\iota_{\hat{Q}_{\mathcal N}} \Theta ,\] 
\item If there is a cohomological vector field $Q_M$ that admits a Hamiltonian function $S_M$, then the transgression of $S_M$ together with $\hat{S}$ defines a solution of the Classical Master Equation, i.e defining the transgression $\breve{S}= \mu_{*} \mbox {ev}^* S_M,$ then 
\begin{equation}\label{CME}
(\hat{S}+ \breve{S},\hat{S}+ \breve{S})=0.
\end{equation}
\end{enumerate}
\end{thm}
\begin{proof}
For the proof of the first claim it is enough to verify that $\Lie_{\hat{Q}_\mathcal{N}}\Theta=0$ but this is consequence of the $Q_\mathcal{N}$-invariance of the Berezinian $\mu$. This last fact is general due to transgression, and it is proven in \cite[Lemma~2.6]{CattaneoFelder3}. The second claim is equivalent to show that 
\begin{eqnarray}
(\hat {\mathcal S}, \hat {\mathcal S})&=&0,\\
(\breve {\mathcal S}, \breve {\mathcal S})&=&0,\\
(\hat {\mathcal S}, \breve {\mathcal S})&=&0.
\end{eqnarray}
Since $Q_N$ is Hamiltonian, it follows that 
\[(\iota_{\hat{Q}_{\mathcal N}} \Theta,\iota_{\hat{Q}_{\mathcal N}} \Theta)=\iota_{[\hat Q_{\mathcal N}, \hat Q_{\mathcal N}] }\Theta=0.\] Therefore $(\hat {\mathcal S}, \hat {\mathcal S})=0$. 
\newline
Since $S_M$ is Hamiltonian, it follows directly that 
\[(\mu_{*} \mbox {ev}^* S_M,\mu_{*} \mbox {ev}^* S_M)=\mu_{*} \mbox {ev}^* (S_M,S_M)=0.\] Therefore 
$(\breve {\mathcal S}, \breve {\mathcal S})=0$. 
\newline
Since $\hat Q$ and $\breve Q$ commute, it follows that $(\hat S, \breve S)$ is constant (see \cite{CattaneoFelder3} for details). Furthermore, since $\mu$ is $Q_{\mathcal N}$-invariant, by Lemma 2.6 in \cite{CattaneoFelder3} it follows that
\[(\hat {\mathcal S}, \breve {\mathcal S})=\mathcal L_{\hat Q}\mu_*\mbox{ev}^*=0.\]
\end{proof}


The next step is to prove that the action functional of the poly-Poisson Sigma model (PPSM), rewritten in the AKSZ-formalism, is a solution of the classical master equation from Theorem \ref{PCME}. For this we must begin with a Lemma that will allow us to get a transgression for any poly--Poisson structure:


\begin{lem}\label{QP_Target}
Let $S$ be a sub-bundle of $\oplus_rT^*M$. Then $S[1]$ is a sub-bundle of $\oplus_rT^*[1]M$. Furthermore, $P$ determines a cohomological vector field on $S[1]$ if and only if $(S,P)$ is $r$-Poisson.
\end{lem}

\begin{proof}
This follows from a standard fact on vector bundles, just by noticing that $(S,P)$ is poly-Poisson if and only if $P:\to TM$ is an anchor map inducing a Lie algebroid structure in $S$.
\end{proof}

Now, we will consider the transgression with $\mathcal{N}=T[1]\partial \Sigma$ and $\mathcal{M}=S[1]$. Since the mapping space $(\mathrm{Map}(\mathcal{N},\oplus_rT[1]^*M),\omega)$ is graded poly-symplectic, then by integrating over the 1-degree fibers of $S[1]\overset{\tau}{\to}\oplus_rT[1]^*M$ we obtain that $\mathcal{F}:=(Map(\mathcal{N},S[1]),\omega_\tau:=\tau^*\omega)$ is also graded poly-symplectic. Moreover, as $\omega$ is exact then also $\omega_\tau=d\theta_\tau$, where $\theta_\tau:=\tau^*\theta$. We define $Q_\mathcal{N}$ to be the cohomological vector field $d$ on $\mathcal{N}$ and $Q_\mathcal{M}$ the vector field induced by $P$. Under this identification we get the following result:
\begin{prop}
$Q_\mathcal{F}:=\hat{Q}_\mathcal{N} +\breve{Q}_\mathcal{M}$ is a cohomological vector field on $\mathcal{F}$.
\end{prop}

Now, similar to the proof of Theorem~\ref{CME}, it is a straightforward verification that the function $\hat{S}=-i_{Q_\mathcal{N}}\Theta$ with $\Theta:=\mathbb T\theta_\tau$ has $Q_\mathcal{N}$  as Hamiltonian vector field with respect to $\Omega:=\mathbb T\omega_\tau$. If in addition, the vector field $Q_{\mathcal{M}}$ is Hamiltonian for $S_P$, that is, $dS_P=i_{Q_\mathcal{M}}\omega_\tau$, we get the following result:
\begin{prop}
$\breve{Q_\mathcal{M}}$ is hamiltonian vector field of $\breve{S}_P:=\mu_*\mathrm{ev}^*S_P$ and $Q_\mathcal{F}$ is hamiltonian for $S_\mathcal{F}:=\hat{S}+\breve{S}_P$, i.e.
$$d\breve{S}_P=i_{\breve{Q}_\mathcal{M}}\Omega \mbox{\ \ and \ \ }dS_{\mathcal{F}}=i_{Q_\mathcal{F}}\Omega.$$
\end{prop}
Finally, by a similar argument as in Section~\ref{sec:PP}, we can define a graded $\mathbb{R}^r$-valued bracket on admissible functions on $\mathcal{F}$. In this case we have $Q_\mathcal{F}|_{\mathrm{Adm}}=(S_\mathcal{F},\cdot)$ and in particular $$Q_\mathcal{F}|_{\mathrm{Adm}}^2=(S_\mathcal{F},(S_\mathcal{F},\cdot)).$$ By the graded Jacobi identity, we can verify the ``poly-version" of the Classical Master Equation:
\begin{thm}
From a poly-Poisson structure $(S,P)$ on $M$ the function $S_\mathcal{F}$ satisfies CME, i.e $(S_\mathcal{F},S_\mathcal{F})=0$.
\end{thm}
The proof follows from the previous results, in particular the fact that $Q_\mathcal{F}$ is cohomological when $P:S\to TM$ defines a Lie algebroid structure, i.e. when $(S,P)$ is poly--Poisson.

Now we are able to justify Equation \eqref{PPSM_Action} from the AKSZ perspective. 
We define the total space of fields of PPSM as 
\begin{equation}
\F^{PPSM}=\mbox{Map}(T[1]\Sigma, T[1]^*M \otimes \mathbb R^r),
\end{equation}
and the space of boundary fields as
\begin{equation}
\F^{PPSM}_{\partial}=\mbox{Map}(T[1]\partial \Sigma, T[1]^*M \otimes \mathbb R^r).
\end{equation}
A straightforward computation shows that
\begin{equation}\label{Transgression_Source}
\hat{S}^{PPSM}= \langle \bm{\eta},  d\bm{X}\rangle
\end{equation}
and
\begin{equation}\label{Transgression_Target}
\breve{S}^{PPSM}= \frac 1 2 \langle P(\bm{X})\bm{\eta}, \bm{\eta} \rangle.
\end{equation}
Combining equations \eqref{Transgression_Source} and \eqref{Transgression_Target} we get 
\begin{equation}\label{PPSM_AKSZ_Action}
 \mathcal S^{PPSM}=\hat{S}^{PPSM}+\breve{S}^{PPSM}= \int_{\Sigma} \langle \bm{\eta}, dX \rangle + \frac 1 2 \langle P(\bm{X})\bm{\eta}, \bm{\eta} \rangle,
\end{equation}
where the pairing takes values in the vector space $\mathbb{R}^r$, and we recover the action functional in Equation \eqref{PPSM_Action}, in degree 0. Furthermore, by using Theorem \ref{PCME}
our desired result follows, namely:
\begin{thm}\label{thm:polyCME}
$\mathcal S^{PPSM}$ is a solution of the classical master equation (as in Theorem \ref{PCME}).
\end{thm}


\begin{thebibliography}{99}
\bibitem{AKSZ} M. Alexandrov, M. Kontsevich, A. Schwarz, O. Zaboronsky, \emph{ The geometry of the master equation and 
topological quantum field theories}, Internat. J. Mod. Phys. A12, 1405-1430, 1997.
\bibitem{Canarutto} D. Canarutto, \emph{Covariant formulation of Lagrangian field theory}, International Journal of Geometric Methods in Modern Physics (2018)
\bibitem{Caratheodory} C. Carath\'eodory, {\it Uber die Extremalen und geod\"atischen Felder in der Variationsrechnung der mehrfachen Integrale}, Acta Sci. Math. (Szeged) {\bf 4} (1929), 193--216.
\bibitem{CattaneoContreras} A. Cattaneo and I. Contreras, {\sl Relational symplectic groupoids}, Letters in Mathematical Physics, Volume 105, {\bf 5}, (2015), pp 723$-$767.
\bibitem{CattaneoFelder} A.S. Cattaneo and G. Felder, \emph{ Poisson sigma models and symplectic groupoids}, in Quantization of Singular Symplectic Quotients, (ed. N. P. Landsman, M. Pflaum, M. Schlichenmeier), Progress in Mathematics 198,Birkh\"auser, 61-93 (2001).
\bibitem{CattaneoFelder2} A. S. Cattaneo and G. Felder, \emph{A path integral approach to the Kontsevich quantization formula}, Commun. Math. Phys. 212 (2000), pp. 591--611.

\bibitem{CattaneoFelder3} A. S. Cattaneo and G. Felder. \emph{On the AKSZ formulation of the Poisson sigma model}, Lett. Math. Phys. 56.2 (2001), pp. 163--179.
\bibitem{CattaneoSchaetz} A. Cattaneo and F. Sch\"atz, \emph{Introduction to Supergeometry}, Rev. Math. Phys. Vol. 23, No. 06, pp. 669--690 (2011)

\bibitem{IvanThesis} I. Contreras, \emph{ Relational Symplectic Groupoids and Poisson Sigma Models with Boundary}, arXiv:1306.3943, PhD Thesis, preprint (2013).
\bibitem{PPSM} I. Contreras and N. Martinez Alba, \emph{Poly-Poisson Sigma Models and Their Relational Poly-Symplectic Groupoids},  Journal of Mathematical Physics, Vol. 59. Issue 7 (2018).
\bibitem{deDonder} T. de Donder, {\it Th\'eorie invariante du calcul des variations}, Nuov.  ́\'ed. (Gauthiers--Villars), Paris 1935.
\bibitem {Finet} C. Finet, L.Quarta, C. Troestler,\emph{Vector-valued variational principles}, Nonlinear Analysis: Theory, Methods and Applications (2003) 
\bibitem{FoGo} M. Forger, L. Gomes, {\it Multisymplectic and Polysymplectic structures on fiber bundles,}
Rev. Math. Phys. {\bf 25} No.9 (2013).

\bibitem{Gunther}C. Gunther, \textit{The polysymplectic Hamiltonian formalism in field theory and calculus of variations. I. The local case,} J. Differential Geom. Volume \textbf{25}.
\bibitem{IMV} D. Iglesias, J.C Marrero, M. Vaquero, {\it Poly-Poisson Structures,} Lett. Math. Phys. Volume 103 (2013) 1103-1133.
\bibitem{Manin} Yu. Manin, \textit{Gauge fields and complex geometry}, Springer-Verlag, Berlin, 1997
\bibitem{NMA-lmp}N. Martinez, {\it Poly-symplectic groupoids and Poly-Poisson structures,} Lett. Math. Phys., Volume 105, (2015), 693--721.
\bibitem{NMA} N. Mart\'inez Alba, {\it On higher Poisson and higher Dirac
structures}, PhD Thesis, IMPA, 2015.
\bibitem{Roy} D. Roytenberg, {\it On the structure of graded symplectic supermanifolds and Courant algebroids}, Quantization, Poisson Brackets and Beyond, Theodore Voronov (ed.), Contemp. Math., Vol. 315, Amer. Math. Soc., Providence, RI, 2002

\end{thebibliography}
\end{document}